\documentclass[12pt,reqno,fleqn]{amsart}
\usepackage{latexsym}
\usepackage{amssymb}
\usepackage{amsxtra}
\usepackage{amsmath}
\usepackage{amsfonts}
\usepackage{CJK}
\newtheorem{theorem}{Theorem}[section]
\newtheorem{lemma}[theorem]{Lemma}
\newtheorem{proposition}[theorem]{Proposition}

\newcommand{\pa}{\partial}

\usepackage{amscd}
\usepackage{amssymb}
\usepackage{mathrsfs}
\usepackage{amsmath}
\usepackage{enumerate}

\makeatother       

\makeatletter      
\@addtoreset{equation}{section}
\makeatother       

\begin{document}

\title[ ]{The Recursion operators of the BKP hierarchy and the CKP Hierarchy}
\author{Maohua Li$^{1,2}$ Jipeng Cheng$^3$  Chuanzhong Li$^2$  Jingsong
He$^2$$^*$} \dedicatory {
1. Department of Mathematics, USTC, Hefei, 230026 Anhui, P.R.China\\
2. Department of Mathematics, NBU, Ningbo, 315211 Zhejiang, P.R.China\\
3. Department of Mathematics,  CUMT, Xuzhou, 221116 Jiangsu,
 P.R.China }
\date{}

\thanks{$^*$ Corresponding author: hejingsong@nbu.edu.cn}

\begin{abstract}

In this paper, under the  constraints of the BKP(CKP) hierarchy, a
crucial observation is that the odd dynamical variable $u_{2k+1}$
can be explicitly expressed by the even dynamical variable $u_{2k}$
in the Lax operator $L$ through a new operator $B$. Using operator
$B$, the essential differences between the BKP hierarchy and the CKP
hierarchy are given by the flow equations and the recursion
operators under the $(2n+1)$-reduction. The formal formulas of the
recursion operators for the BKP and CKP hierarchy under
$(2n+1)$-reduction are given. To illustrate this method, the two
recursion operators are constructed explicitly for the 3-reduction
of the BKP and CKP hierarchies. The $t_7$ flows of $u_2$ are
generated from $t_1$ flows by the above recursion operators, which
are consistent with the corresponding flows generated by the flow
equations under $3$-reduction.
\end{abstract}
\maketitle
Mathematics Subject Classifications(2000).  37K05, 37K10, 37K40.\\
\textbf{Keywords}: BKP hierarchy, CKP hierarchy, Flow equation,
Recursion operator.

\section{Introduction}

The Kadomtsev-Petviashvili (KP) hierarchy \cite{dick1} is  an
attractive research object in the mathematical physics since 1980s.
It has two important sub-hierarchies, i. e.  Kadomtsev-Petviashvili
hierarchy of B-type (BKP hierarchy) and Kadomtsev-Petviashvili
hierarchy of C-type (CKP hierarchy) \cite{olver2} which are two
interesting reductions of KP hierarchy associated with two infinite
dimensional algebras $o({\infty})$ and $sp({\infty})$ respectively
\cite{jimbo93}. By the formulation of pseudo-differential operator,
the Lax operator of KP hierarchy is
$L=\sum_{l\geq0}u_l\partial^{1-l}=\partial+u_2\partial^{-1}+u_3\partial^{-2}+\cdots$.
These different algebraic structures also can be shown in some sense
by the reduction conditions on the Lax operator, i. e.  $L^*=-L$ for
the CKP hierarchy and  $L^*=-\partial L\partial^{-1}$ for the BKP
hierarchy \cite{jimbo93}. Here $L^*$ is the formally adjoint
operator of $L$. Besides above two essential differences between the
KP, BKP and CKP hierarchies from the view of algebra structure and
Lax operator, there are also more interesting facts in the four
aspects including the $\tau$ function \cite{jimbo93}, the gauge
transformation \cite{csy1}-\cite{hwc1}, the additional symmetries
and the ASvM formula \cite{OS86}-\cite{tm2011}, the freeze of the
even flows and  odd number dynamical variables $u_j,
(j=3,5,7,\cdots)$ \cite{jimbo93}. Because of the importance of the
flow equations and recursion operators, it is very natural to
explore more differences from these two aspects among them.

The recursion operator  \cite{olver77} for a given soliton equation
is firstly introduced by using the KdV equation as example. The
results of the recent thirty five years of the  soliton theory show
possessing a recursion operator is one of essential integrable
properties, which is related to infinitely many conservation laws
and symmetries, Hamiltonian structure, higher order flows, etc
\cite{olver2,ovel1,sokolov}. With the help of the recursion
operator, the higher flows can be generated from the lower flows for
an integrable hierarchy, which offers a natural way to construct the
whole  hierarchy from a single seed system
\cite{olver2,sokolov,loris99,lcz2011,boiti1,fokas1,santini1,cjp1}.
There are several ways  to construct the recursion operator of a
given integrable system, which is reviewed  in reference
\cite{sokolov}. The construction of recursion operator in 2+1
dimension was given in the papers by Fokas and Santini
\cite{fokas1,santini1,fokas2}. For the $n$-reduction KP hierarchy,
it is very natural to extract recursion operator from the explicit
flow equations, which has been done by W. Strampp and W. Oevel
\cite{ovel1}. The advantage of this method is that the higher-order
flows generated by recursion operator from lower-order ones are
local even if the recursion operator has nonlocal term, because the
higher-order flows are automatically identified with the local flows
given by the Lax equations of the KP hierarchy.

To improve the understanding of more essential differences of the
KP, BKP and CKP hierarchies, we shall study  the explicit flows and
the recursion operators for them.  The main difficulties to apply
Strampp and Oevel's method \cite{ovel1} for the BKP hierarchy and
the CKP hierarchy are due to the two constraints: the disappearance
of the even flows and the odd dynamical variables
$u_i,(i=1,3,5,\cdots)$, which is originated from the reduction
conditions on the Lax operator $L^*=-\partial
L\partial^{-1}$ or $L^*=-L$. It is also crucial to know that there only exists
$(2n+1)$-th reduction in the BKP hierarchy and the CKP hierarchy.
The key step is to transmit the reduction conditions on the Lax
operator to the flow equations, which can be realized by expressing
the odd dynamical variable $u_j,(j=1,3,5,\cdots)$ by the even ones.

This paper was organized as follows. The odd dynamical variable
$u_j,(j=1,3,5,\cdots,)$ is expressed by the even ones in section 2
with the reduction condition on Lax operator $L^*=-L$ or
$L^*=-\partial L\partial^{-1}$. An operator $B$  is introduced for
this purpose. In section 3, the odd flow equations of the even
dynamical variable $u_{2j},(j=1,2,\cdots,)$ are obtain from
$L_{t_{2m+1}}=[L,(L^{2m+1})_-].$ We also calculate the some odd
flows of BKP(CKP) hierarchy as examples. In section 4 the recursion
operator of BKP(CKP) hierarchy is discussed. This recursion
operators are different from the recursion operator of KP hierarchy
\cite{ovel1,sokolov,cjp1}. Section 5 is devoted on reduction of  the
$t_7$ flows of BKP(CKP) hierarchy by the recursion operator, and
which are consistent with corresponding flows in section 2.  Section
6 is a brief discussion  on recursion operator and the future
problem.

\section{The Even dynamical variables }
In this section, we study the even dynamical variables and the odd
ones of the Lax operator. The even dynamical variables will be
expressed by the odd dynamical variables by a formula which be
introduced below.

Firstly we given a pseudo-differential Lax operator
\begin{eqnarray}
L&=&\sum_{l\geq0}u_l\partial^{1-l}\nonumber\\
&=&\partial+u_2\partial^{-1}+u_3\partial^{-2}+\cdots ,\label{lax}
\end{eqnarray}
where we assume $u_0=1,u_1=0$, and $u_2,u_3,\cdots$ are the
functions of an infinite set of time variables
$t=(t_1=x,t_3,\cdots)$ and $\partial=\frac{\partial}{\partial x}$,
by imposing the following condition on $L$
\begin{eqnarray}
L^*=-\partial^k L \partial^{-k}\label{bckpcondition}.
\end{eqnarray}
And the formally adjoined operators are given by
\begin{eqnarray}
L^*=\sum_{l\geq0}{(-1)}^{1-l}\partial^{1-l}u_l,
\end{eqnarray}
with $k=0,1$ corresponding to the Lax operator for CKP hierarchy and
BKP hierarchy respectively. The operation of $\partial^k$ with $k\in
\mathbb{Z}$ is defined by
\begin{eqnarray}
\partial^ku=\sum_{j\geq0}C_k^ju^{(j)}\partial^{k-j},
\end{eqnarray} where $u^{(j)}=\frac{\partial^j u}{\partial x^j}$, with
\begin{eqnarray}
C_k^j=\frac{k(k-1)(k-2)\cdots(k-j+1)}{j!},j\in \mathbb{Z_+}.
\end{eqnarray}
The $m$-th power of $L$ can be denote for
\begin{eqnarray}
L^m&=&{(\partial+u_2\partial^{-1}+u_3\partial^{-2}+\cdots)}^m\nonumber\\
&=&\sum_{j\leq m} p_j(m)\partial^j
\end{eqnarray}
and  $(L^m)_+=\sum^m_{j=0}p_j(m)\partial^j$, i. e.  $(L^m)_+$ is the
non-negative projection of $L^m$, and $(L^m)_-=L^m-(L^m)_+$ is the
negative projection of $L^m$. In particular,
\begin{equation}\label{uprelation}
    u_l=p_{1-l}(1), l=2,3,\cdots.
\end{equation}
The dynamical equations of BKP hierarchy and CKP hierarchy are defined as
follows,
\begin{equation}\label{bckphierarchy}
    L_{t_{2m+1}}=[(L^{2m+1})_{+},L]=[L,(L^{2m+1})_-],\quad m=0,1,2,\cdots
\end{equation}
Remembering the corresponding constraints on the Lax operator, there
are only the odd flows existed with the BKP hierarchy and CKP
hierarchy \cite{jimbo93}. And  the $(2n+1)$-th power of Lax operator
$L$ must be considered,
\begin{equation}\label{nlax}
L^{2n+1}=\sum_{j\leq 2n+1}
p_j(2n+1)\partial^j,
\end{equation}
with $n=0,1,2,\cdots$. From the constraints (\ref{bckpcondition}),
we know
\begin{equation}\label{nckpcondtion}
\partial^{-k} L^{*2n+1}=-L^{2n+1}\partial^{-k}.
\end{equation}
By considering the negative part of both sides
\begin{equation}\label{negativenbckpcondtion}
(\partial^{-k}L^{*2n+1})_{-}=-(L^{2n+1}\partial^{-k})_{-},
\end{equation}
then
\begin{eqnarray}
l.h.s~\rm{of}~(\ref{negativenbckpcondtion})&=&\sum_{j \leq
-1+k}(-1)^{j}\partial
^{j-k} p_j(2n+1)\nonumber\\
&=&\sum_{j \leq -1+k}\sum_{l\geq 0}(-1)^{j}C_{j-k}^l
p_j^{(l)}(2n+1)\partial
^{j-k-l}\nonumber\\
&=&\sum_{j \leq-1+k}\sum_{l\geq 0}^ {-j-1+k}(-1)^{j+l}C_{j+l-k}^l
p_{j+l}^{(l)}(2n+1)\partial^{j-k}.\nonumber
\end{eqnarray}
On the other hand,
\begin{eqnarray}
r.h.s~\rm{of}~(\ref{negativenbckpcondtion})&=&-\sum_{j \leq
-1+k}p_j(2n+1)\partial^{j-k} .\nonumber
\end{eqnarray}
So
\begin{equation}\label{npowerckp}
\sum_{j\leq-1+k}\sum_{l\geq 0}^ {-j-1+k}(-1)^{j+l}C_{j+l-k}^l
p_{j+l}^{(l)}(2n+1)\partial^{j-k}=-\sum_{j\leq
-1+k}p_j(2n+1)\partial^{j-k}.
\end{equation}
Comparing the coefficients of $\partial^{j-k}$ in above relation, we
find
\begin{equation}\label{coeffnbckp}
\sum_{l\geq 0}^ {-j-1+k}(-1)^{j+l}C_{j+l-k}^l
p_{j+l}^{(l)}(2n+1)=-p_j(2n+1),\quad j\leq -1+k.
\end{equation}
Further,
\begin{equation}\label{recoeffnckp}
((-1)^j+1)p_j(2n+1)=-\sum_{l\geq 1}^ {-j-1+k}(-1)^{j+l}C_{j+l-k}^l
p_{j+l}^{(l)}(2n+1),\quad j\leq -1+k.
\end{equation}
Thus, one can find $p_j(2n+1)$ with $j$ odd are independent. As for
$j$ being even number, we have
\begin{equation}\label{evenrecoeffnbckp}
p_j(2n+1)=-\frac{1}{2}\sum_{l\geq 1}^ {-j-1+k}(-1)^{j+l}C_{j+l-k}^l
p_{j+l}^{(l)}(2n+1),\quad j\leq -1+k~\rm{and}~ \emph{j}~\rm{even}.
\end{equation}
In particular, $p_0(2n+1)=0$ for $k=1$. Thus $p_j(2n+1)$ in
(\ref{evenrecoeffnbckp}) becomes
\begin{equation}\label{newevenbcncoeffient}
   p_j(2n+1)=-\frac{1}{2}\sum_{l\geq 1}^ {-j-1}(-1)^{l}C_{j+l-k}^l
p_{j+l}^{(l)}(2n+1),\quad j=-2,-4,\cdots.
\end{equation}
By considering (\ref{uprelation}), we have  $l$ odd number dynamical
variable
\begin{eqnarray*}
    u_l&=&p_{1-l}(1)\\
    &=&-\frac{1}{2}\sum_{\mu\geq 1}^
    {l-2}(-1)^{1-l+\mu}C_{1-l+\mu-k}^\mu
p_{1-l+\mu}^{(\mu)}(1)\\
 &=&-\frac{1}{2}\sum_{\mu\geq 1}^
    {l-2}(-1)^{\mu}C_{1-l+\mu-k}^\mu
u_{l-\mu}^{(\mu)}, \quad l=3,5,\cdots.
\end{eqnarray*}
That is,
\begin{equation}\label{coeffbckplax}
    u_l=-\frac{1}{2}\sum_{\mu\geq 1}^
    {l-2}(-1)^{\mu}C_{1-l+\mu-k}^\mu
u_{l-\mu}^{(\mu)}=-\frac{1}{2}\sum_{\mu\geq 1}^
    {l-2}C_{l-2+k}^\mu
u_{l-\mu}^{(\mu)},\quad l=3,5,\cdots .
\end{equation}

So we summarize above results for below. For BKP hierarchy,
\begin{eqnarray}
    u_l&=&-\frac{1}{2}\sum_{\mu\geq 1}^{l-2}C_{l-1}^\mu
u_{l-\mu}^{(\mu)},\quad l=3,5,\cdots,\label{oddcoeffibkplax}\\
p_j(2n+1)&=&-\frac{1}{2}\sum_{l\geq 1}^ {-j-1}(-1)^{l}C_{j+l-1}^l
p_{j+l}^{(l)}(2n+1),\quad j= -2,-4,\cdots.\label{evencoeffibkpnlax}
\end{eqnarray}
For CKP hierarchy,
\begin{eqnarray}
    u_l&=&-\frac{1}{2}\sum_{\mu\geq 1}^{l-2}C_{l-2}^\mu
u_{l-\mu}^{(\mu)},\quad l=3,5,\cdots,\label{oddcoeffickplax}\\
p_j(2n+1)&=&-\frac{1}{2}\sum_{l\geq 1}^ {-j-1}(-1)^{l}C_{j+l}^l
p_{j+l}^{(l)}(2n+1),\quad j=-2,-4,\cdots.\label{evencoeffickpnlax}
\end{eqnarray}

But because the odd dynamical variable and even dynamical variable
of it are not really separated by  eq.(\ref{oddcoeffibkplax},
\ref{oddcoeffickplax}), we next want to separate the odd parts and
even ones from above relation. Before doing this, let's see a lemma
first.
\begin{lemma}\label{oddeven}
If
\begin{equation}\label{lemmacondition}
    p_j=\sum_{\mu\geq 1}^{-j-1}A_{j\mu}p_{j+\mu},\quad
    j=-2,-4,-6,\cdots,
\end{equation}
where $A_{j\mu}$ is an operator, then
\begin{equation}\label{lemmaresult1}
    p_{-2l}=\sum_{\mu=1}^lB_{-2l,-2\mu+1}p_{-2\mu+1}
\end{equation}
with
\begin{equation}\label{lemmaresult2}
    B_{-2l,-2\mu+1}=\sum_{\scriptstyle i_1+\cdots+i_\nu\leq 2l-2\mu
    , \nu\geq0 \atop \scriptstyle
    i_\gamma \rm{is~ positive~ even~ number}}(\prod_{\gamma=1}^\nu
    A_{-2l+i_1+\cdots+i_{\gamma-1},i_\gamma})A_{-2l+i_1+\cdots+i_{\nu},-2\mu+1+2l-i_1-i_2-\cdots-i_\nu}.
\end{equation}
\end{lemma}
\begin{proof} We prove the lemma by induction. Obviously the lemma is true
for $l=1$. We next assume the lemma is correct for $\leq l$, then
for $l+1$ case. By (\ref{lemmacondition}),
\begin{eqnarray*}
p_{-2l-2}&=&\sum_{\gamma\geq
1}^{2l+1}A_{-2l-2,\gamma}p_{-2l-2+\gamma}=\sum_{\mu=
1}^{l+1}A_{-2l-2,2\mu-1}p_{-2l-2+2\mu-1}+\sum_{\mu=
1}^{l}A_{-2l-2,2\mu}p_{-2l-2+2\mu}\\
&=&\sum_{\mu=
1}^{l+1}A_{-2l-2,-2\mu+1+2(l+1)}p_{-2\mu+1}+\sum_{\gamma=
1}^{l}A_{-2l-2,2\gamma}p_{-2(l-\gamma+1)}\\
&=&\sum_{\mu= 1}^{l+1}A_{-2l-2,-2\mu+1+2(l+1)}p_{-2\mu+1}\\
&&+\sum_{\mu= 1}^{l+1}\sum_{\gamma=
1}^{l}A_{-2l-2,2\gamma}A_{-2l-2+2\gamma,i_2}A_{-2l-2+2\gamma+i_2,i_3}\cdots\\
&&A_{-2l-2+2\gamma+i_2+\cdots+i_{\nu-1},i_\nu}A_{-2l-2+2\gamma+i_2+\cdots+i_{\nu},-2\mu+1+2(l+1)-2\gamma-i_2-\cdots-i_\nu}p_{-2\mu+1}\\
&=&\sum_{\mu= 1}^{l+1} B_{-2l-2,-2\mu+1}p_{-2\mu+1}.
\end{eqnarray*}
Thus the lemma holds for $l+1$ case.\end{proof}

We apply this lemma to BKP hierarchy and CKP hierarchy cases,
\begin{equation}\label{lemmaapply}
    {A_{j\mu}} =
    -\frac{1}{2}(-1)^{\mu}C_{j+\mu-k}^\mu\pa^\mu=-\frac{1}{2}C_{-j+k-1}^\mu\pa^\mu
=\left\{ {\begin{array}{*{20}{c}}
   { - \frac{1}{2}C_{-j}^\mu{\partial ^\mu},k=1, BKP }, \\
   {- \frac{1}{2}C_{-j-1}^\mu{\partial ^\mu},k=0, CKP}.\\
\end{array}} \right.
\end{equation}
So the corresponding $B$ operators are,
\begin{eqnarray}
&&{B_{-2l,-2\mu + 1}} \nonumber\\
&=&\mathop\sum \limits_{{i_1},{i_2}, \cdots ,{i_\nu }} {( -
\frac{1}{2})^{\nu  + 1}}C_{2l - 1 + k}^{{i_1}}C_{2l - 1 + k- {i_1}
}^{{i_2}} \cdots C_{2l - 1 + k- {i_1} - {i_2} -  \cdots  - {i_{\nu -
1}} }^{{i_\nu }}C_{2l - 1 + k - {i_1} -  \cdots  - {i_\nu }}^{ -
2\mu  + 1 + 2l - {i_1} -  \cdots  - {i_\nu }}{\partial ^{2l - 2\mu +
1}}\nonumber\\
 &= &\left\{ {\begin{array}{*{20}{c}}
   {\sum\limits_{{i_1},\cdots,{i_\nu }} {{{( - \frac{1}{2})}^{\nu  + 1}}C_{2l}^{{i_1}}C_{2l - {i_1}}^{{i_2}} \cdots C_{2l - {i_1} - {i_2} -  \cdots  - {i_{\nu  - 1}}}^{{i_\nu }}C_{2l - {i_1} -  \cdots  - {i_\nu }}^{ - 2\mu + 1 + 2l - {i_1} -  \cdots  - {i_\nu }}{\partial ^{2l - 2\mu + 1}},\quad\quad\quad BKP,} }  \\
   {\sum\limits_{{i_1},\cdots,{i_\nu }} {{{( - \frac{1}{2})}^{\nu  + 1}}C_{2l - 1}^{{i_1}}C_{2l - 1 - {i_1}}^{{i_2}} \cdots C_{2l - 1 - {i_1} - {i_2} -  \cdots  - {i_{\nu  - 1}}}^{{i_\nu }}C_{2l - 1 - {i_1} -  \cdots  - {i_\nu }}^{ - 2\mu + 1 + 2l - {i_1} -  \cdots  - {i_\nu }}{\partial ^{2l - 2\mu + 1}},CKP.} }  \\
\end{array}} \right.\nonumber\\
\label{lemmaapplyresult}
\end{eqnarray}
So according to Lemma \ref{oddeven}, we get
\begin{equation}\label{lemmaresultforbkpckp}
    p_{-2l}(2n+1)=\sum_{\mu=1}^lB_{-2l,-2\mu+1}p_{-2\mu+1}(2n+1).
\end{equation}

\begin{proposition}\label{oddu(2l+1)}
All the odd dynamical variables $u_{2l+1}$ of Lax operator $L$ can
be expressed by the even dynamical variables $u_{2\mu}(\mu\leq l)$,
that is,
\begin{equation}\label{uformual}
u_{2l+1}=\sum_{\mu=1}^lB_{-2l,-2\mu+1}u_{2\mu},
\end{equation}
where $B_{-2l,-2\mu+1}$ is defined by (\ref{lemmaapplyresult}).
\end{proposition}

\begin{proof}
From eq. (\ref{lemmaresultforbkpckp}), in particular,
$u_{2l+1}=p_{-2l}(1)$ for $l=1,2,\cdots$. Then  we find
\begin{equation}
u_{2l+1} = p_{-2l}(1) =\sum_{\mu=1}^lB_{-2l,-2\mu+1}p_{-2\mu+1}(1)
 = \sum_{\mu=1}^lB_{-2l,-2\mu+1}u_{2\mu}.\nonumber
\end{equation}
\end{proof}

The equation (\ref{uformual}) is crucial to calculate the flow
equation in  section \ref{sectionflow} and the recursion operator in
section \ref{sectionrecursion}. From the above relation of $u_j$,
one can obtain that all the odd item $u_{2l+1}$ can be expressed by
the even item $u_{2\mu}$, where $\mu\leq l$.

Below let's see some examples of (\ref{uformual}). Firstly we will
deal with the BKP hierarchy.

For $l=1$, we find
$\mu=1$,$B_{-2l,-2\mu+1}=B_{-2,-1}=-\frac{1}{2}C_2^1\pa=-\pa$, thus
$u_3=-u_{2x}$.

For $l=2$, then $\mu=1,\mu=2$, thus one only need $B_{-4,-1}$ and
$B_{-4,-3}$, while
\begin{eqnarray*}
B_{-4,-1}&=&-\frac{1}{2}C_4^3\pa^3+(-\frac{1}{2})^2C_4^2C_2^1\pa^3=\pa^3,\\
B_{-4,-3}&=&-\frac{1}{2}C_4^1\pa=-2\pa.
\end{eqnarray*}
So
$$u_5=u_{2xxx}-2u_{4x}.$$

For $l=3$, $\mu=1,2,3$, $B_{-6,-1},B_{-6,-3}$ and $B_{-6,-5}$ have
the form
\begin{eqnarray*}
B_{-6,-1}&=&\{(-\frac{1}{2})C_6^5+(-\frac{1}{2})^2C_6^2C_4^3+(-\frac{1}{2})^2C_6^4C_2^1+(-\frac{1}{2})^3C_6^2C_4^2C_2^1\}\pa^5=-3\pa^5,\\
B_{-6,-3}&=&\{(-\frac{1}{2})C_6^3+(-\frac{1}{2})^2C_6^2C_4^1\}\pa^3=5\pa^3,\\
B_{-6,-5}&=&(-\frac{1}{2})C_6^1\pa=-3\pa.
\end{eqnarray*}
So
\begin{equation*}
   u_7=-3u_{2xxxxx}+5u_{4xxx}-3u_{6x}.
\end{equation*}

For $l=4$ and $\mu=1,2,3,4$,$B_{-8,-1},B_{-8,-3},B_{-8,-5}$ and
$B_{-8,-7}$ have the form
\begin{eqnarray*}
B_{-8,-1}&=&\{(-\frac{1}{2})C_8^7+(-\frac{1}{2})^2C_8^2C_6^5+(-\frac{1}{2})^3C_8^2C_6^2C_4^3+(-\frac{1}{2})^2C_8^4C_4^3+(-\frac{1}{2})^2C_8^6C_2^1\\
&+&(-\frac{1}{2})^3C_8^2C_6^4C_2^1+(-\frac{1}{2})^3C_8^4C_4^2C_2^1+(-\frac{1}{2})^4C_8^2C_6^2C_4^2C_2^1\}\pa^7=17\pa^7,\\
B_{-8,-3}&=&\{(-\frac{1}{2})C_8^5+(-\frac{1}{2})^2C_8^2C_6^3+(-\frac{1}{2})^2C_8^4C_4^1+(-\frac{1}{2})^3C_8^2C_6^2C_4^1\}\pa^5=-28\pa^5,\\
B_{-8,-5}&=&\{(-\frac{1}{2})C_8^3+(-\frac{1}{2})^2C_8^2C_6^1\}\pa^3=14\pa^3,\\
B_{-8,-7}&=&(-\frac{1}{2})C_8^1\pa=-4\pa.
\end{eqnarray*}
So
\begin{equation*}
   u_9=17u_{2xxxxxxx}-28u_{4xxxxx}+14u_{6xxx}-4u_{8x}.
\end{equation*}

Then we consider the examples of CKP hierarchy.

When $l=1$ and $\mu=1$,
$B_{-2,-1}=-\frac{1}{2}C_1^1\pa=-\frac{1}{2}\pa$. Thus
$u_3=-\frac{1}{2}u_{2x}.$

For $l=2$, then $\mu=1,2$, thus only need $B_{-4,-1}$ and
$B_{-4,-3}$, while
\begin{eqnarray*}
B_{-4,-1}&=&-\frac{1}{2}C_3^3\pa^3+(-\frac{1}{2})^2C_3^2C_1^1\pa^3=\frac{1}{4}\pa^3,\\
B_{-4,-3}&=&-\frac{1}{2}C_3^1\pa=-\frac{3}{2}\pa.
\end{eqnarray*}
So
$$u_5=\frac{1}{4}u_{2xxx}-\frac{3}{2}u_{4x}.$$

For $l=3$, $\mu=1,2,3$, $B_{-6,-1},B_{-6,-3}$ and $B_{-6,-5}$ have
the form
\begin{eqnarray*}
B_{-6,-1}&=&\{(-\frac{1}{2})C_5^5+(-\frac{1}{2})^2C_5^2C_3^3+(-\frac{1}{2})^2C_5^4C_1^1+(-\frac{1}{2})^3C_5^2C_3^2C_1^1\}\pa^5=-\frac{1}{2}\pa^5,\\
B_{-6,-3}&=&\{(-\frac{1}{2})C_5^3+(-\frac{1}{2})^2C_5^2C_3^1\}\pa^3=\frac{5}{2}\pa^3,\\
B_{-6,-5}&=&(-\frac{1}{2})C_5^1\pa=-\frac{5}{2}\pa.
\end{eqnarray*}
So
\begin{equation*}
   u_7=-\frac{1}{2}u_{2xxxxx}+\frac{5}{2}u_{4xxx}-\frac{5}{2}u_{6x}.
\end{equation*}

For $l=4$, $\mu=1,2,3,4$, $B_{-8,-1},B_{-8,-3},B_{-8,-5}$ and
$B_{-8,-7}$ have the  form below
\begin{eqnarray*}
B_{-8,-1}&=&\{(-\frac{1}{2})C_7^7+(-\frac{1}{2})^2C_7^2C_5^5+(-\frac{1}{2})^3C_7^2C_5^2C_3^3+(-\frac{1}{2})^2C_7^4C_3^3+(-\frac{1}{2})^2C_7^6C_1^1,\\
&+&(-\frac{1}{2})^3C_7^2C_5^4C_1^1+(-\frac{1}{2})^3C_7^4C_3^2C_1^1+(-\frac{1}{2})^4C_7^2C_5^2C_3^2C_1^1\}\pa^7=\frac{17}{8}\pa^7,\\
B_{-8,-3}&=&\{(-\frac{1}{2})C_7^5+(-\frac{1}{2})^2C_7^2C_5^3+(-\frac{1}{2})^2C_7^4C_3^1+(-\frac{1}{2})^3C_7^2C_5^2C_3^1\}\pa^5=-\frac{21}{2}\pa^5,\\
B_{-8,-5}&=&\{(-\frac{1}{2})C_7^3+(-\frac{1}{2})^2C_7^2C_5^1\}\pa^3=\frac{35}{4}\pa^3,\\
B_{-8,-7}&=&(-\frac{1}{2})C_7^1\pa=-\frac{7}{2}\pa.
\end{eqnarray*}
So
\begin{equation*}
   u_9=\frac{17}{8}u_{2xxxxxxx}-\frac{21}{2}u_{4xxxxx}+\frac{35}{4}u_{6xxx}-\frac{7}{2}u_{8x}.
\end{equation*}
We summarize above results below.

For BKP,
\begin{eqnarray}\label{bkpconstrain}
\begin{cases}
u_3=-u_{2x},\\
u_5=u_{2xxx}-2u_{4x},\\
u_7=-3u_{2xxxxx}+5u_{4xxx}-3u_{6x},\\
u_9=17u_{2xxxxxxx}-28u_{4xxxxx}+14u_{6xxx}-4u_{8x},\\
\cdots
\end{cases}
\end{eqnarray}

For CKP,
\begin{eqnarray}\label{ckpconstrain}
\begin{cases}
u_3=-\frac{1}{2}u_{2x}\\
u_5=\frac{1}{4}u_{2xxx}-\frac{3}{2}u_{4x},\\
u_7=-\frac{1}{2}u_{2xxxxx}+\frac{5}{2}u_{4xxx}-\frac{5}{2}u_{6x},\\
u_9=\frac{17}{8}u_{2xxxxxxx}-\frac{21}{2}u_{4xxxxx}+\frac{35}{4}u_{6xxx}-\frac{7}{2}u_{8x},\\
\cdots
\end{cases}
\end{eqnarray}

\textbf{Remark:} From (\ref{uformual}), one can know there are only
the even dynamical variables of $\{u_j,j\geq1\}$ are independent,
and the odd dynamical variables of $\{u_j,j\geq1\}$ can be expressed
by the even ones of $\{u_j,j\geq1\}$. With this result, it is nature
to discuss the odd flows of even dynamical variables in the next
section.

\section{Flow Equations}\label{sectionflow}
 We next deal with the BKP hierarchy and CKP hierarchy in an unified
 way. First we derive the flow equations of the dynamical variables
 $u_{2j}$. Inserting (\ref{lax}) and (\ref{nlax}) into
(\ref{bckphierarchy}), one finds
\begin{eqnarray*}
&&L(L^{2m+1})_--(L^{2m+1})_-L\\
&=&\sum_{r\geq
0}\sum_{h>0}(u_r\pa^{1-r}p_{-h}(2m+1)\pa^{-h}-p_{-h}(2m+1)\pa^{-h}u_r\pa^{1-r})\\
&=&\sum_{r\geq 0}\sum_{\alpha\geq 0}\sum_{h>0}(C_{1-r}^\alpha
u_rp_{-h}^{(\alpha)}(2m+1)-C_{-h}^\alpha
p_{-h}(2m+1)u_r^{(\alpha)})\pa^{1-r-h-\alpha}\\
&=&\sum_{l\geq 0}\sum_{r\geq 0}^l\sum_{h>0}(C_{1-r}^{l-r}
u_rp_{-h}^{(l-r)}(2m+1)-C_{-h}^{l-r}
p_{-h}(2m+1)u_r^{(l-r)})\pa^{1-l-h}\\
&=&\sum_{j\geq 1}\sum_{h\geq 1}^j\sum_{r\geq
0}^{j-h}(C_{1-r}^{j-h-r} u_rp_{-h}^{(j-h-r)}(2m+1)-C_{-h}^{j-h-r}
p_{-h}(2m+1)u_r^{(j-h-r)})\pa^{1-j}.
\end{eqnarray*}
Comparing with $L_{t_{2m+1}}=\sum_{j\geq 0}u_{j,t_{2m+1}}\pa^{1-j}$,
we have
\begin{equation}\label{flowequationbc}
    u_{0,t_{2m+1}}=0,\quad u_{1,t_{2m+1}}=0,\quad
    u_{j,t_{2m+1}}=\sum_{h=1}^j O_{j,h}p_{-h}(2m+1),
\end{equation}
where
\begin{equation}\label{oexpression}
    O_{j,h}=\sum_{r\geq
0}^{j-h}(C_{1-r}^{j-h-r} u_r\pa^{j-h-r}-C_{-h}^{j-h-r}
u_r^{(j-h-r)}).
\end{equation}
In particular, $O_{j,j}=0,O_{j,j-1}=\pa.$

With (\ref{lax}) and (\ref{nlax}), $p_j(2n+1)$ can be uniquely
determined by $u_2, u_3, \cdots, u_{2n+1-j}$, i. e.  it's formula is
\begin{equation}\label{pformula}
    p_j(2n+1)=(2n+1)u_{2n+1-j}+f_{jn}(u_2,u_3, \cdots, u_{2n-j}),j\leq
    2n+1
\end{equation} and $f_{jn}$ is a differential polynomials in $u_2,u_3, \cdots,
u_{2n-j}$. With the help of Proposition \ref{oddu(2l+1)}, every
dynamical variable can be expressed by the even ones. So $p_j(2n+1)$
can be expressed by $u_2,u_4,\cdots,u_{2n+1-j}$ which $j$ is odd.
Now we consider the $(2n+1)$-reduction, i. e.  for some fixed $2n+1,
n \in \mathbb{Z_+}$,
\begin{equation}\label{nreduction}
    L^{2n+1}=(L^{2n+1})_+.
\end{equation}
This relation is equal to requiring the $p_j(2n+1)=0$ for $j<0$.
Hence, one can recursively express all coordinates $u_j$ with $j\geq
2n+1$ in terms of $(u_2,u_3, \cdots, u_{j-1})$. But thanks for
(\ref{uformual}), all the odd dynamical variables $u_{2l+1}$ can be
express by the even dynamical variables $u_{2\mu}$, where $\mu\leq
l$. So only first $n$ coordinates $(u_2,u_4,\cdots,u_{2n})$ are
independent for BKP(CKP) hierarchy in the case of
$(2n+1)$-reduction.

On the other hand, according to formula (\ref{lemmaresultforbkpckp})
and the third formula of (\ref{flowequationbc}), one has
\begin{eqnarray}
u_{2j,t_{2m+1}}&=&\sum_{h=1}^{j}O_{2j,2h-1}p_{-2h+1}(2m+1)+\sum_{h=1}^{j}O_{2j,2h}p_{-2h}(2m+1)\nonumber\\
&=&
\sum_{h=1}^{j}O_{2j,2h-1}p_{-2h+1}(2m+1)+\sum_{h=1}^{j}\sum_{\mu=1}^hO_{2j,2h}B_{-2h,-2\mu+1}p_{-2\mu+1}(2m+1)\nonumber\\
&=&
\sum_{h=1}^{j}O_{2j,2h-1}p_{-2h+1}(2m+1)+\sum_{\mu=1}^{j}\sum_{h=\mu}^jO_{2j,2h}B_{-2h,-2\mu+1}p_{-2\mu+1}(2m+1)\nonumber\\
&=&\sum_{h=1}^{j}\Big(O_{2j,2h-1}+\sum_{\mu=h}^jO_{2j,2\mu}B_{-2\mu,-2h+1}\Big)p_{-2h+1}(2m+1).\nonumber
\end{eqnarray}
Thus,
\begin{equation}\label{evenflowequation}
    u_{2j,t_{2m+1}}=\sum_{h=1}^{j}Q_{jh}p_{-2h+1}(2m+1), \quad j\leq n,
\end{equation}
which
\begin{equation}\label{qexpression}
Q_{jh}=O_{2j,2h-1}+\sum_{\mu=h}^jO_{2j,2\mu}B_{-2\mu,-2h+1},\quad
j\leq n.
\end{equation}
Obviously, $Q_{jj}=\pa$. So with the above formula
(\ref{evenflowequation}) and the help of equations
(\ref{bkpconstrain}, \ref{ckpconstrain}), all odd flow equations of
the even dynamical coordinate $u_{2j}$ can be obtained. This result
implies that the flow equations are expressed by even dynamical
variable ${u_2,u_4,\cdots,u_{2n}}$. Due to the appearance of the
operator $B$ (\ref{lemmaapplyresult}), the flow equation of KP
hierarchy, BKP hierarchy and CKP hierarchy are different.

We present some odd flow equations below calculated by Maple. The
first several odd flow equations of BKP hierarchy are
\begin{eqnarray}\label{bkpu2tflow}
\begin{cases}
u_{2,t_1}=u_{2,x},\\
u_{2,t_3}=6u_{2}u_{2,x}+3u_{4,x}-2u_{2,xxx},\\
u_{2,t_5} =
20u_2u_{4,x}+20u_4u_{2,x}+10u_2u_{2,xxx}+5u_{6,x}+60u_{2,x}u_{2,xx}-\frac{2}{3}u_{2,xxxxx}+30u_2^2u_{2,x},\\
u_{2,t_7} =
\frac{10}{3}u_{2,xxxxxxx}+406u_{2,x}u_2u_{2,xx}+210u_2u_4u_{2,x}+7u_{8,x}-7u_{4,xxxxx}\\
\quad\quad +14u_{6,xxx}+112u_{2,x}^3+42u_2u_{6,x}+42u_4u_{4,x}+42u_6u_{2,x}+49u_2u_{4,xxx}\\
\quad\quad
+98u_4u_{2,xxx}+49u_2u_{2,xxxxx}+203u_{2,x}u_{4,xx}+252u_{4,x}u_{2,xx}+294u_{2,x}u_{2,xxxx}\\
\quad\quad
+609u_{2,xx}u_{2,xxx}+105u_2^2u_{4,x}+91u_2^2u_{2,xxx}+140u_2^3u_{2,x}.
\end{cases}
\end{eqnarray}

The first several odd flow equations of CKP hierarchy are
\begin{eqnarray}\label{ckpu2tflow}
\begin{cases}
u_{2,t_1}= u_{2,x},\\
u_{2,t_3}= 6u_2u_{2,x}+3u_{4,x}-\frac{1}{2}u_{2,xxx},\\
u_{2,t_5}=20u_2u_{4,x}+20u_4u_{2,x}+5u_{6,x}+95u_{2,x}u_{2,xx}+30u_2u_{2,xxx}\\
\quad\quad +10u_{4,xxx}+30u_2^2u_{2,x}-\frac{3}{2}u_{2,xxxxx},\\
u_{2,t_7}=42u_2u_{6,x}+42u_4u_{4,x}+42u_6u_{2,x}+49u_2u_{4,xxx}+98u_4u_{2,xxx}+35u_2u_{2,xxxxx}\\
\quad\quad +7u_{8,x}+\frac{385}{2}u_{2,x}u_{4,xx}+\frac{483}{2}u_{2,xx}u_{4,x}+154u_{2,x}u_{2,xxxx}+287u_{2,xx}u_{2,xxx}\\
\quad\quad
+105u_2^2u_{4,x}+91u_2^2u_{2,xxx}+140u_2^3u_{2,x}+210u_2u_4u_{2,x}+\frac{791}{2}u_{2,x}u_2u_{2,xx}\\
\quad\quad
+14u_{6,xxx}+\frac{427}{4}u_{2,x}^3-7u_{4,xxxxx}+\frac{13}{6}u_{2,xxxxxxx}.\\
\end{cases}
\end{eqnarray}

One can see the first flow equations (\ref{bkpu2tflow},
\ref{ckpu2tflow}) of BKP(CKP) hierarchy
 are trivial equations. If
consider the $3$-reduction (\ref{nreduction}) when $n=1$, we can
calculate the $t_7$ flow  from the $t_1$ flow. The  first three
equations of $3$-reduction  of BKP hierarchy are
\begin{eqnarray}\label{bkp3reduction}
\begin{cases}
u_4 = -u_2^2+\frac{2}{3}u_{2,xx},\\
u_6
=-2u_2u_4-\frac{11}{3}u_{2,x}^2-\frac{1}{3}u_2^3-\frac{7}{3}u_2u_{2,xx}+\frac{2}{3}u_{4,xx}-\frac{1}{3}u_{2,xxxx},\\
u_8
 =-2u_2u_6-\frac{7}{3}u_2u_{4,xx}+9u_2u_{2,x}^2-\frac{1}{3}u_2^2u_{2,xx}-\frac{32}{3}u_{2,x}u_{4,x}-10u_4u_{2,xx}\\
\quad\quad +\frac{7}{3}u_{2,xx}^2+\frac{2}{3}u_{6,xx}
-\frac{1}{3}u_{4,xxxx}+\frac{1}{9}u_{2,xxxxxx}-u_4^2-u_2^2u_4.
\end{cases}
\end{eqnarray}
If one  substitute (\ref{bkp3reduction}) into (\ref{bkpu2tflow}),
then the $t_7$ flow equation of $u_2$ (\ref{bkpu2tflow}) of BKP
hierarchy can be reduced for
\begin{eqnarray}\label{bkpt7flow}
u_{2,t_7}=
&-&\frac{7}{9}u_2u_{2,xxxxx}-\frac{14}{9}u_{2,x}u_{2,xxxx}-\frac{7}{3}u_{2,xx}u_{2,xxx}-\frac{14}{3}u_{2,xxx}u_{2}^2\nonumber\\
&-&\frac{28}{3}u_2^3u_{2,x}-14u_{2,x}u_2u_{2,xx}-\frac{7}{3}u_{2,x}^3-\frac{1}{27}u_{2,xxxxxxx}.
\end{eqnarray}

The first three   equations of $3$-reduction of CKP hierarchy are
\begin{eqnarray}\label{ckp3reduction}
\begin{cases}
u_4 = -u_2^2+\frac{1}{6}u_{2,xx},\\
u_6 =-2u_2u_4-\frac{17}{12}u_{2,x}^2-\frac{1}{3}u_2^3-\frac{4}{3}u_2u_{2,xx}+\frac{2}{3}u_{4,xx}-\frac{1}{6}u_{2,xxxx},\\
u_8 =-2u_2u_6+\frac{9}{4}u_2u_{2,x}^2-\frac{49}{6}u_{2,x}u_{4,x}+\frac{2}{3}u_{2,x}u_{2,xxx}-u_4^2-\frac{7}{3}u_2u_{4,xx}-\frac{1}{3}u_2u_{2,xxxx}\\
\quad\quad
-\frac{4}{3}u_2^2u_{2,xx}-8u_{4}u_{2,xx}+\frac{17}{12}u_{2,xx}^2+\frac{2}{3}u_{6,xx}
-\frac{1}{3}u_{4,xxxx}+\frac{1}{18}u_{2,xxxxxx}-u_2^2u_4.\\
\end{cases}
\end{eqnarray}
If one  substitute (\ref{ckp3reduction}) into (\ref{ckpu2tflow}),
then the $t_7$ flow equation of $u_2$ (\ref{ckpu2tflow}) of CKP
hierarchy can be reduced for
\begin{eqnarray}\label{ckpt7flow}
u_{2,t_7}=
&-&\frac{35}{6}u_{2,x}^3-\frac{28}{3}u_2^3u_{2,x}-\frac{14}{3}u_2^2u_{2,xxx}-\frac{14}{3}u_{2,xx}u_{2,xxx}-21u_{2,x}u_2u_{2,xx}\nonumber\\
&-&\frac{49}{18}u_{2,x}u_{2,xxxx}-\frac{1}{27}u_{2,xxxxxxx}-\frac{7}{9}u_2u_{2,xxxxx}.
\end{eqnarray}

But it is not easy to find a relation between the more other higher
order flow equations and the lower order flow equations. And we will
find the recursion operator which can generate the higher order flow
equations from the lower order flow equations in the next section.

For $(2n+1)$-reduction, it only has the odd reduction and even
dynamical variable in the BKP(CKP) hierarchy. If we denote
\begin{eqnarray*}
 \widehat{U}(2n) &=& {({u_2},{u_4}, \cdots ,{u_{2n}})^t}, \\
 \widehat{P}(2n+1,2m+1) &= &{({p_{ - 1}}(2m + 1),{p_{ - 3}}(2m + 1), \cdots ,{p_{ - 2n + 1}}(2m + 1))^t}, \\
 Q(n) &=& \left( {\begin{array}{*{20}{c}}
   {{Q_{11}}} & 0 &  \cdots  & 0  \\
   {{Q_{21}}} & {{Q_{22}}} &  \cdots  & 0  \\
    \vdots  &  \vdots  &  \ddots  &  \vdots   \\
   {{Q_{n1}}} & {{Q_{n2}}} &  \cdots  & {{Q_{nn}}}  \\
\end{array}} \right),
\end{eqnarray*} where the up index $t$ denotes the transpose of the matrix,
then (\ref{evenflowequation}) can be rewritten for
\begin{equation}\label{nflowequation}
    \widehat{U}(2n)_{t_{2m+1}}=Q(n)\widehat{P}(2n+1,2m+1).
\end{equation}
It is trivial to know that all the flow equations in
$\widehat{U}(2n)_{t_{2m+1}}$ are local. Next, we want to study the
recursion relation between $t_{2m+1+2p(2n+1)}$ flow and $t_{2m+1}$
flow.

\section{Recursion Operator}\label{sectionrecursion}
In this section, we will discuss the recursion operator of BKP
hierarchy and CKP hierarchy starting from the recursion operator of
KP hierarchy.  To do this, we must find a recursion formula relation
between $\widehat{U}(2n)_{t_{2m+1}}$ and
$\widehat{U}(2n)_{t_{2m+4n+3}}$ under the $(2n+1)$-reduction
constraint. That is, we try to find an operator
$\widehat{\Phi}(2n+1)$, s.t.
$\widehat{U}(2n)_{t_{2m+4n+3}}=\widehat{\Phi}(2n+1)\widehat{U}(2n)_{t_{2m+1}}$.
Recall the result of the recursion operator of KP hierarchy
\cite{ovel1,cjp1} under $n$-reduction, we have
\begin{equation}
    P(n,m+n)=R(n)P(n,m),
\end{equation}
where
\begin{eqnarray*}
P(n,m)&=&(p_{-1}(m),p_{-2}(m),\cdots,p_{-n+1}(m))^t,\\
R(n)&=&S(n)-T(n)M(n)^{-1}N(n),
\end{eqnarray*}
\begin{eqnarray*}
S(n)& =& \left( {\begin{array}{*{20}{c}}
   {{C_{- 1,0}}(n)} & {{C_{ - 1,1}}(n)} &  \cdots  & {{C_{ - 1,n - 2}}(n)}  \\
   {{C_{ - 2, - 1}}(n)} & {{C_{ - 2,0}}(n)} &  \cdots  & {{C_{ - 2,n - 3}}(n)}  \\
    \vdots  &  \vdots  &  \ddots  &  \vdots   \\
   {{C_{ - n + 1 , - n + 2}}(n)} & {{C_{ - n + 1 , - n + 3}}(n)} &  \cdots  & {{C_{ - n + 1 ,0}}(n)}. \\
\end{array}} \right)_{(n-1)\times(n-1)},
\end{eqnarray*}
\begin{eqnarray*}
T(n) &= &\left( {\begin{array}{*{20}{c}}
   {{C_{- 1,n - 1}}(n)} & {{C_{ - 1,n}}(n)} & 0 &  \cdots  & 0  \\
   {{C_{- 2,n - 2}}(n)} & {{C_{- 2,n - 1}}(n)} & {{C_{ - 2,n}}(n)} &  \cdots  & 0  \\
    \vdots  &  \vdots  &  \vdots  &  \ddots  &  \vdots   \\
   {{C_{ - n + 1 ,1}}(n)} & {{C_{ - n + 1 ,2}}(n)} & {{C_{ - n + 1 ,3}}(n)} &  \cdots  & {{C_{ - n + 1 ,n}}(n)}\\
\end{array}} \right)_{(n-1)\times n},
\end{eqnarray*}
\begin{eqnarray*}
M(n)& =  &\left( {\begin{array}{*{20}{c}}
   { - n\partial  }& 0 &  \cdots  & 0  \\
   {{D_{ - 2 ,n - 2}}(n)} & { - n\partial }&  \cdots  & 0  \\
    \vdots  &  \vdots  &  \ddots  &  \vdots   \\
   {{D_{ - n ,0}}(n)} & {{D_{ - n ,1}}(n)} &  \cdots  & { - n\partial }\\
\end{array}} \right)_{n\times n},
\end{eqnarray*}

\begin{eqnarray*}
 N(n)&= & \left( {\begin{array}{*{20}{c}}
   {{D_{ - 1 ,0}}(n)} & {{D_{ - 1 ,1}}(n)} &  \cdots  & {{D_{ - 1 ,n - 3}}(n)} & {{D_{ - 1 ,n - 2}}(n)}  \\
   {{C_{ - 2 , - 1}}(n)} & {{D_{ - 2,0}}(n)} &  \cdots  & {{D_{ - 2 ,n - 4}}(n)} & {{D_{ - 2 ,n - 3}}(n)}  \\
    \vdots  &  \vdots  &  \ddots  &  \vdots  &  \vdots   \\
   {{C_{ - n + 1 , - n + 2}}(n)} & {{C_{ - n + 1 , - n + 3}}(n)} &  \cdots  & {{C_{ - n + 1 , - 1}}(n)} & {{D_{ - n + 1 ,0}}(n)}  \\
   {{C_{ - n , - n + 1}}(n)} & {{C_{ - n , - n + 2}}(n)} &  \cdots  & {{C_{ - n, - 2}}(n)} & {{C_{ - n , - 1}}(n)}  \\
\end{array}} \right)_{n\times(n-1) },
\end{eqnarray*}

\begin{eqnarray*}
C_{j,\mu}(n)&=&\sum_{l=max(0,\mu)}^n
C_{j-\mu}^{l-\mu}p_l^{(l-\mu)}(n),\\
D_{j,s}&=&C_{j,s}(n)-\widetilde{C}_{s}(n),\\
\widetilde{C}_s(n)&=&\sum_{\mu=
0}^{n-s}C_{s+\mu}^sp_{s+\mu}(n)\pa^\mu.
\end{eqnarray*}

If we set $\Phi(n) = Q(n)R(n)Q^{-1}(n)$, then the recursion formula
of  KP hierarchy  \cite{ovel1,cjp1} is
\begin{equation}\label{kprecursion}
U(n)_{t_{m+jn}}=\Phi^j(n)U(n)_{t_m},
\end{equation}
 where
$U(n)=(u_2,u_3,u_4,\cdots,u_{n-1},u_n)^t$. If  we  substitute $2m+1$
for $m$, $2n+1$ for  $n$ and   $j=2$ in (\ref{kprecursion}),  we
have
\begin{equation}\label{bckprecursion}
U(2n+1)_{t_{2m+1+2(2n+1)}}=\Phi^{2}(2n+1)U(2n+1)_{t_{2m+1}}.
\end{equation}
We consider the even element of $U(2n+1)$, which are the dynamical
variables of BKP(CKP) hierarchy. Then it is necessary to calculate
the odd flow equations of the even dynamical variables
$u_{2k,t_{2m+1+2p(2n+1)}}$. And if let $\Phi(n)_{i,j}$ denote the
$(i,j)$-th element of the matrix $\Phi(n)$, from
(\ref{bckprecursion}), the $(2k-1)$-th elements of
$U(2n+1)_{t_{2m+1+2(2n+1)}}$ are
\begin{eqnarray}\label{phi}
\lefteqn{u_{2k,t_{2m+1+2(2n+1)}}=\sum_{i=1}^{2n}(\Phi^{2}(2n+1))_{2k-1,i}u_{i+1,t_{2m+1}}}\nonumber\\
&=&\sum_{i=1}^{n}(\Phi^{2}(2n+1))_{2k-1,2i-1}u_{2i,t_{2m+1}}+\sum_{i=1}^{n}(\Phi^{2}(2n+1))_{2k-1,2i}u_{2i+1,t_{2m+1}}\nonumber\\
&=&\sum_{i=1}^{n}(\Phi^{2}(2n+1))_{2k-1,2i-1}u_{2i,t_{2m+1}}+\sum_{i=1}^{n}{\sum_{\mu=1}^{i}(\Phi^{2}(2n+1))_{2k-1,2i}B_{-2i,-2\mu+1}u_{2\mu,t_{2m+1}}}\nonumber\\
&=&\sum_{i=1}^{n}(\Phi^{2}(2n+1))_{2k-1,2i-1}u_{2i,t_{2m+1}}+\sum_{\mu=1}^{n}{\sum_{i=\mu}^{n}(\Phi^{2}(2n+1))_{2k-1,2i}B_{-2i,-2\mu+1}u_{2\mu,t_{2m+1}}}\nonumber\\
&=&\sum_{\mu=1}^{n}[(\Phi^{2}(2n+1))_{2k-1,2\mu-1}+\sum_{i=\mu}^{n}(\Phi^{2}(2n+1))_{2k-1,2i}B_{-2i,-2\mu+1}]u_{2\mu,t_{2m+1}},k\leq
n.
\end{eqnarray}
It is used the formula (\ref{uformual}) for the third equality. If
denote
$\widehat{\Phi}(2n+1)_{k,\mu}=(\Phi^{2}(2n+1))_{2k-1,2\mu-1}+\sum_{i=\mu}^{n}(\Phi^{2}(2n+1))_{2k-1,2i}B_{-2i,-2\mu+1}$,
 then (\ref{phi}) become
\begin{eqnarray}\label{phibar1}
u_{2k,t_{2m+1+2(2n+1)}}=\sum_{\mu=1}^{n}\widehat{\Phi}(2n+1)_{k,\mu}u_{2\mu,t_{2m+1}}.
\end{eqnarray}
Further we denote
\begin{eqnarray}\label{recursionoperator22}
\widehat{\Phi}(2n+1)&=&(\widehat{\Phi}(2n+1)_{k,\mu})\nonumber\\
&=&((\Phi^{2}(2n+1))_{2k-1,2\mu-1}+\sum_{i=\mu}^{n}(\Phi^{2}(2n+1))_{2k-1,2i}B_{-2i,-2\mu+1}),
\end{eqnarray} and $\widehat{\Phi}(2n+1)$ is a $n\times n$ matrix
because $1\leq k\leq n$ and $1\leq \mu\leq n$. Then for
$\widehat{U}(2n)=({u_2},{u_4}, \cdots ,u_{2n})^t$, one has
\begin{eqnarray}\label{recursionformula}
\widehat{U}(2n)_{t_{2m+1+2(2n+1)}}=\widehat{\Phi}(2n+1)\widehat{U}(2n)_{t_{2m+1}}.
\end{eqnarray}

With the above prepared knowledge, we  have a theorem below.
\begin{theorem}\label{recursionth}
The flow equations of BKP(CKP) hierarchy under the
$(2n+1)$-reduction
 possess a recursion
operator $\widehat{\Phi}(2n+1)$ such that
\begin{equation}\label{recursion}
\widehat{U}(2n)_{t_{2m+1+2p(2n+1)}}=\widehat{\Phi}^p(2n+1)\widehat{U}(2n)_{t_{2m+1}},
\end{equation} where $\widehat{\Phi}(2n+1)$ is defined by
(\ref{recursionoperator22}).
\end{theorem}

\begin{proof}
With  (\ref{recursionoperator22}) and (\ref{recursionformula}), it
is clear that we have
\begin{eqnarray*}
\widehat{U}(2n)_{t_{2m+1+2p(2n+1)}}&=&\widehat{U}(2n)_{t_{2(m+(p-1)(2n+1))+1+2(2n+1)}}\\
&=&\widehat{\Phi}(2n+1)\widehat{U}(2n)_{t_{2(m+(p-1)(2n+1))+1}}\\
&=&\widehat{\Phi}(2n+1)\widehat{U}(2n)_{t_{2m+1+2(p-1)(2n+1)}}\\
&&\cdots\\
&=&\widehat{\Phi}^p(2n+1)\widehat{U}(2n)_{t_{2m+1}}.
\end{eqnarray*}
\end{proof}

\textbf{Remark:} Under the $(2n+1)$-reduction,
$t_1,t_3,\cdots,t_{2n-1},t_{2n+3},t_{2n+5},\cdots,t_{4n+1} $-flows
are independent, and only n coordinates $(u_2,u_4,\cdots,u_{2n}) $
are independent. That is just the $2n$ flows  can generate the whole
BKP(CKP) hierarchy under the action of the recursion operator
$\widehat{\Phi}(2n+1)$ (\ref{recursionoperator22}). Though the
recursion operator $\widehat{\Phi}(2n+1)$ is nonlocal, but it
doesn't generate the nonlocal higher flow equations. Because the
flow equations (\ref{bckphierarchy}) are local, and the recursion
operator $\widehat{\Phi}(2n+1)$ is derived from these flow
equations. In particular, the difference  of the recursion operators
in eq.(\ref{recursion}) of the BKP hierarchy and the CKP hierarchy
is reflected by the appearance of the  operator $B$.

\section{Applications}
In this section, we will give some examples for the applications of
formula (\ref{recursion}). Here we only consider $3$-reduction of
the BKP and CKP hierarchies. For an example, we generate the $t_7$
flow equation from the $t_1$ flow equation for $3$-reduction.

For the BKP hierarchy, set $n=1,m=0$ and $p=1$ in (\ref{recursion}),
one can calculate
\begin{eqnarray}\label{phibar2}
\Phi(3)=\left(
\begin{array}{cc}
\Phi_{11}(3)&\Phi_{12}(3)\\
\Phi_{21}(3)&\Phi_{22}(3)
\end{array}
\right),
\end{eqnarray}
where
\begin{eqnarray*}
\begin{cases}
\Phi_{11}(3)=\frac{1}{3}\partial^3+\frac{1}{3}a_1\partial-\frac{1}{3}a_{1,x}-\frac{1}{3}a_{1,xx}\partial^{-1},\\
\Phi_{12}(3)=\frac{2}{3}\partial^2+\frac{2}{3}a_1+\frac{1}{3}a_{1,x}\partial^{-1},\\
\Phi_{21}(3)=-\frac{2}{9}\partial^4-\frac{4}{9}a_1\partial^2-\frac{2}{3}a_{1,x}\partial-\frac{2}{9}a_1^2+(\frac{1}{9}a_{1,xxx}-\frac{2}{9}a_1a_{1,x})\partial^{-1},\\
\Phi_{22}(3)=-\frac{1}{3}\partial^3-\frac{1}{3}a_1\partial-a_{1,x}-\frac{1}{3}a_{1,xx}\partial^{-1},
\end{cases}
\end{eqnarray*} and $a_1(3)=3u_2,a_0(3)=0$.
Because $B_{-2,-1}=-\partial$, then the recursion operator is
\begin{eqnarray}\label{hatphi}
\widehat{\Phi}(3)&=&\Phi_{11}^2(3)+\Phi_{12}(3)\Phi_{21}(3)-(\Phi_{11}(3)\Phi_{12}(3)+\Phi_{12}(3)\Phi_{22}(3))\partial\nonumber\\
&=&-\frac{1}{27}\partial^6-\frac{2}{3}u_2\partial^4-u_{2,x}\partial^3-(\frac{11}{9}u_{2,xx}+3u_2^2)\partial^2-(\frac{10}{9}u_{2,xxx}+7u_2u_{2,x})\partial\nonumber\\
&-&(\frac{5}{9}u_{2,xxxx}+2u_{2,x}^2+4u_2^3+\frac{16}{3}u_2u_{2,xx})-u_{2,x}\partial^{-1}(\frac{2}{3}u_{2,xx}+u_2^2)\nonumber\\
&-&(\frac{1}{9}u_{2,xxxxx}+\frac{5}{3}u_2u_{2,xxx}+\frac{5}{3}u_{2,x}u_{2,xx}+5u_2^2u_{2,x})\partial^{-1}.
\end{eqnarray} With the recursion operator
(\ref{recursionoperator22}), we can generate $t_7$ flow from $t_1$
flow
 by $u_{2,t_7}=\widehat{\Phi}(3)u_{2,t_1}=\widehat{\Phi}(3)u_{2,x}$, i. e.
 \begin{eqnarray}\label{bkpu7flow}
u_{2,t_7}=
&-&\frac{7}{9}u_2u_{2,xxxxx}-\frac{14}{9}u_{2,x}u_{2,xxxx}-\frac{7}{3}u_{2,xx}u_{2,xxx}-\frac{14}{3}u_{2,xxx}u_{2}^2\nonumber\\
&-&\frac{28}{3}u_2^3u_{2,x}-14u_{2,x}u_2u_{2,xx}-\frac{7}{3}u_{2,x}^3-\frac{1}{27}u_{2,xxxxxxx},
\end{eqnarray}  which consistent with the flow eq. (\ref{bkpt7flow}) of the BKP hierarchy under $3$-reduction.
With a scaling transformations for $u_2\rightarrow \frac{u}{3}$ and
$t_7\rightarrow -27t$, the operator (\ref{hatphi}) consistent with
the formula (B3) of Ref. \cite{sokolov} , and (\ref{bkpu7flow})
become the flow equation
\begin{eqnarray}
u_{t}&=&3u_{xxxxxxx}+15uu_{xxxxx}+189u_{xxx}^3+1134uu_xu_{xx}\nonumber\\
&+&126u_xu_{xxxx}
+756u^3u_x+189u_{xx}u_{xxx}+126u^2u_{xxx}\label{bkpu7flow2}.
\end{eqnarray}

Set $n=1,m=0$ and $p=1$ in eq.(\ref{recursion}) for the CKP
hierarchy, then
\begin{eqnarray}
\Phi(3)=\left(
\begin{array}{cc}
\Phi_{11}(3)&\Phi_{12}(3)\\
\Phi_{21}(3)&\Phi_{22}(3)
\end{array}
\right),
\end{eqnarray}
where
\begin{eqnarray*}
\begin{cases}
\Phi_{11}(3)=\frac{1}{3}\partial^3+\frac{1}{3}a_1\partial-\frac{1}{3}a_{1,x}+a_0(3)+(\frac{2}{3}a_{0,x}-\frac{1}{3}a_{1,xx})\partial^{-1},\\
\Phi_{12}(3)=\frac{2}{3}\partial^2+\frac{2}{3}a_1+\frac{1}{3}a_{1,x}\partial^{-1},\\
\Phi_{21}(3)=-\frac{2}{9}\partial^4-\frac{4}{9}a_1\partial^2-\frac{2}{3}a_{1,x}\partial-\frac{2}{9}a_1^2-\frac{2}{3}a_{0,x}+(\frac{1}{9}a_{1,xxx}-\frac{2}{9}a_1a_{1,x}-\frac{1}{3}a_{0,xx})\partial^{-1},\\
\Phi_{22}(3)=-\frac{1}{3}\partial^3-\frac{1}{3}a_1\partial-a_{1,x}+a_0+\frac{1}{3}(a_{0,x}-a_{1,xx})\partial^{-1},
\end{cases}
\end{eqnarray*} and $a_1(3)=3u_2$, $a_0(3)=\frac{3}{2}u_{2,x}$.
Because $B_{-2,-1}=-\frac{1}{2}\partial$, then the recursion
operator is
\begin{eqnarray}\label{hatphi2}
\widehat{\Phi}(3)&=&\widehat{\Phi}(3)_{1,1}\nonumber\\
&=&\Phi_{11}^2(3)+\Phi_{12}(3)\Phi_{21}(3)-\frac{1}{2}(\Phi_{11}(3)\Phi_{12}(3)+\Phi_{12}(3)\Phi_{22}(3))\partial\nonumber\\
&=&-\frac{1}{27}\partial^6-\frac{2}{3}u_2\partial^4-2u_{2,x}\partial^3-(\frac{49}{18}u_{2,xx}+3u_2^2)\partial^2-(\frac{35}{18}u_{2,xxx}+10u_2u_{2,x})\partial\nonumber\\
&-&(\frac{13}{18}u_{2,xxxx}+\frac{41}{6}u_2u_{2,xx}+\frac{23}{4}u_{2,x}^2+4u_2^3)-\frac{1}{6}u_{2,x}\partial^{-1}(u_{2,xx}+6u_2^2)\nonumber\\
&-&(\frac{1}{9}u_{2,xxxxx}+\frac{5}{3}u_2u_{2,xxx}+\frac{25}{6}u_{2,x}u_{2,xx}+5u_2^2u_{2,x})\partial^{-1}.
\end{eqnarray}
We can generate $t_7$ flow from $t_1$ flow by
$u_{2,t_7}=\widehat{\Phi}(3)u_{2,t_1}=\widehat{\Phi}(3)u_{2,x}$,
 i. e.
\begin{eqnarray}\label{ckpu7flow2}
u_{2,t_7}=
&-&\frac{35}{6}u_{2,x}^3-\frac{28}{3}u_2^3u_{2,x}-\frac{14}{3}u_2^2u_{2,xxx}-\frac{14}{3}u_{2,xx}u_{2,xxx}-21u_{2,x}u_2u_{2,xx}\nonumber\\
&-&\frac{49}{18}u_{2,x}u_{2,xxxx}-\frac{1}{27}u_{2,xxxxxxx}-\frac{7}{9}u_2u_{2,xxxxx},
\end{eqnarray} which consistent with  flow eq.(\ref{ckpt7flow}) for the  CKP hierarchy.  With a scaling transformations $u_2\rightarrow
\frac{3}{2}u$ and $t_7\rightarrow  -27t$, (\ref{hatphi2}) is nothing
but  the formula (30) of Ref. \cite{sokolov} and (\ref{ckpu7flow2})
become the flow of equation
\begin{eqnarray}
u_{t}&=&u_{xxxxxxx}+14uu_{xxxxx}+49uu_{xxxx}+84u_{xx}u_{xxx}+56u^2u_{xxx}\nonumber\\
&+&\frac{224}{3}u^3u_x+256uu_{x}u_{xx}+70u_{x}^3.\label{ckpu7flow3}
\end{eqnarray}

\textbf{Remark:} If let $n=1,m=1$ and $p=1$ in (\ref{recursion}), we
can also obtain the second recursion relation by
$\widehat{\Phi}(3)$, i. e. the formula
$\widehat{U}(2)_{t_{11}}=\widehat{\Phi}(3)\widehat{U}(2)_{t_5}$. It
is not difficult to generate $t_{11}$ flow equation of $u_2$ from
$t_5$ flow equation of  it. If we choose properly the number
 of coefficient $n,m,p$  in (\ref{recursion}), then
$\widehat{U}(2)_{t_{13}}=\widehat{\Phi}^2(3)\widehat{U}(2)_{t_1}$ is
obtained for  $n=1,m=0$ and $p=2$, and
$\widehat{U}(2)_{t_{17}}=\widehat{\Phi}^2(3)\widehat{U}(2)_{t_5}$ is
obtained  for $n=1$ and $m=2=p$. And the highest order of $\partial$
in
 $\widehat{\Phi}^2(3)$ is 12. Of course one can also use it to
generate the higher order flows.

\section{Conclusions and Discussions}

 In this paper, we found in Proposition \ref{oddu(2l+1)} that the odd dynamical variable $u_{2k+1}$ of
$\{u_j,j\geq1\}$ can be expressed by the even dynamical variable
$u_{2k}$ of $\{u_j,j\geq1\}$ in Lax operator $L$ by considering the
constraint of BKP(CKP) hierarchy. The flow equations and the
recursion operators of the BKP and the CKP hierarchies are given in
a unified approach, which also reflect  the two essential
differences between the two sub-hierarchies of the KP hierarchy
because of the appearance of the operator $B$.  Two examples of the
recursion operator are given explicitly for the BKP and CKP
hierarchy under the $3$-reduction. The $t_7$ flows are generated by
these recursion operators again, which are consistent with the flow
equations in the Lax equation.  So the validity of these recursion
operators is confirmed.

This research depicts deeply  the  the integrability   of KP
hierarchy and BKP(CKP) hierarchy. And it will also be helpful for
studying the difference of  Hamiltonian  structure, Poisson bracket
between the  KP hierarchy and BKP(CKP) hierarchy, which  will be
studied later.

{\bf Acknowledgments} {\noindent \small  This work is supported by
the NSF of China under Grant No.10971109 and Science Fund in Ningbo
University (No.xk1062, No.XYL11012). Jingsong He is also supported
by Program for NCET under Grant No.NCET-08-0515.


\begin{thebibliography}{99}

\bibitem{dick1}
L. A. Dickey, Soliton Equations and Hamiltonian Systems, (2nd edn.
World Scientific, Singapore, 2003).

\bibitem{olver2}
P. J. Olver, Applications of Lie groups to differential equations,
(New York: Springer, 1993).

\bibitem{jimbo93}
M. Jimbo, and T. Miwa, Solitions and infinite dimensional lie
algebras, Publ. RIMS, Kyoto Univ. 19(1983), 943-1001.


\bibitem{csy1}
L.-L. Chau, J.-C. Shaw and  H.-C. Yen, Solving the KP hierarchy by
Gauge Transformations, Commun. Math. Phys. 149(1992), 263-278.

\bibitem{hcr1}
J. S. He, Y. Cheng and A. Rudolf R\"{o}mer, Solving bi-directional
soliton equations in the KP hierarchy by gauge transformation, JHEP
03(2006) 103.



\bibitem{cst1}L.-L. Chau, J.-C. Shaw and M.-H. Tu, Solving the constrained KP hierarchy
by gauge transformations, J. Math. Phys. 38(1997), 4128-4137.


\bibitem{wlg1}  R.  Willox, I.  Loris, C.  R.  Gilson, Binary Darboux transformations for constrained
 KP hierarchies,  Inverse Problems 13(1997), 849-865.


\bibitem{hwc1}
J. S. He, Z. W. Wu and  Y. Cheng, Gauge transformations for the
constrained CKP and BKP hierarchies, J. Math. Phys. 48(2007),
113-519.


\bibitem{OS86}
A.Yu. Orlov and E.I. Schulman, Additional symmetries for integrable
systems and conformal algebra repesentation,  Lett. Math. Phys. 12
(1986), 171-179.

\bibitem{dick95}
L. A. Dickey, On additional symmetries of the KP hierarchy and
Sato's B\"acklund transformation, Comm. Math. Phys. 167(1995),
227-233.


\bibitem{ASM95}
M. Adler, T. Shiota and P. van Moerbeke, A Lax representation for
the vertex operator and the central extension, Comm. Math. Phys. 171
(1995), 547-588.


\bibitem{L95}
J. van de Leur, The Adler-Shiota-van Moerbeke formula for the BKP
hierarchy, J. Math. Phys. 36 (1995), 4940-4951.

\bibitem{tu07}
M. H. Tu,  On the BKP hierarchy: Additional symmetries, Fay identity
and Adler-Shiota- van Moerbeke formula, Lett. Math. Phys. 81(2007),
93-105.


\bibitem{hetian1}
J. S. He, K. L. Tian, A. Foerster and W. X. Ma, Additional
Symmetries and String Equation of the CKP Hierarchy, Lett. Math.
Phys. 81(2007), 119-134.




\bibitem{chh2011}
J. P. Cheng, J. S. He and S. Hu, The "ghost" symmetry of the BKP
hierarchy, J. Math. Phys. 51(2010), 053514  .

\bibitem{tianhe2011}
K. L. Tian, J. S. He, J.P. Cheng  and Y. Cheng, Additional
symmetries of constrained  CKP and BKP hierarchies, Science China
Mathematics 54(2011), 257-268.

\bibitem{tm2011}
H. F. Shen and M.H. Tu, On the constrained B-type
Kadomtsev¨CPetviashvili hierarchy: Hirota bilinear equations and
Virasoro symmetry, J. Math. Phys. 52(2011), 032704.



\bibitem{olver77}
P. J. Olver, Evolution equations possessing infinitely many
symmetries, J. Math. Phys. 18(1977), 1212-1215.



\bibitem{ovel1} W. Strampp and W. Oevel, Recursion operators and
Hamiltonian structures in Sato¡¯s theory, Lett. Math. Phys.
20(1990), 195-210.



\bibitem{sokolov}
M. Gurses, A. Karasu and V. V. Sokolov, On construction of recursion
operators from Lax representation, J. Math. Phys. 40(1999),
6473-6490.

\bibitem{loris99}
I. Loris, Recursion operator for a constraint BKP system, In:Boiti
M, Martina L, etal ed. Proceedings of the Workshop on Nonlinearity,
Integrability and All That Twenty years After NEEDS'79. Singapore:
World Scientific, 1999. 325-330.




\bibitem{lcz2011}
C. Z. Li, K. L. Tian, J. S. He, etal, Recursion operator for a
constrained CKP hierarchy, Acta Mathematica Scientia 31B(2011),no.4,
1295-1302.



\bibitem{boiti1}
M. Boiti, J.JP.Leon, L. Martina and F. Pempinelli, On the recurion
operator for the KP hierarchy in two and three spatial dimensions,
Phys. Lett. A  123(1987), 340-344.


\bibitem{fokas1}
A. S. Fokas and P. M. Santini, The Recursion Operator of the
Kadomtsev-Petviashvili Equation and the Squared Eigenfunction of the
Schr$\ddot{o}$dinger Operators, Stud. Appl. Math. 75(1986), 179-186.

\bibitem{santini1}
P. M. Santini and A. S. Fokas, Recursion operators and
bi-Hamiltonian structures in multidimensions. I, Commun. Math. Phys.
115(1988), 375-419.


\bibitem{cjp1}
J. P. Cheng, L. H. Wang and J. S. He, Resursion operators for KP,
mKP and Harry-Dym hirearchies, J.  Nonlinear Mathematical Physics
18(2011), 161-178.



\bibitem{fokas2}
A. S. Fokas and P. M. Santini, Recursion operators and
bi-Hamiltonian structures in multidimensions. II, Commun. Math.
Phys. 116(1988), 449-474.


\bibitem{sand1}
J. A. Sanders and J. P. Wang, Integrable systems and their recurion
operators, Nonlinear Analysis 47(2001), 5213-5240.

\end{thebibliography}
\end{document}